\documentclass[reqno,10pt]{article} 
\usepackage{amsmath,amsthm,amssymb,amscd,verbatim}
\usepackage{mathtools}
\usepackage{amssymb}
\usepackage{graphicx}
\usepackage{graphicx}
\usepackage{epsfig}
\usepackage{tikz}
\usepackage[ruled]{algorithm2e}

\usepackage{color}

\usepackage{pgf,tikz}
\usepackage{mathrsfs}
\usetikzlibrary{arrows}
\usepackage[final]{pdfpages}
\usepackage{pstricks-add}

\newcommand{\ignore}[1]{}

\input{epsf}

\newtheorem{prelem}{{\bf Theorem}}

\newtheorem{theorem}{Theorem}
\newtheorem{corollary}[theorem]{Corollary}
\newtheorem{definition}[theorem]{Definition}
\newtheorem{conjecture}[theorem]{Conjecture}

\newtheorem{question}[theorem]{Question}
\newtheorem{lemma}[theorem]{Lemma}

\newtheorem{remarka}[theorem]{Remark}

\newtheorem{examplea}[theorem]{Example}

\newtheorem{exercisea}[theorem]{Exercise}

\DeclareMathOperator*{\argmax}{arg\,max}

\def\Ex {{\mathbb E}}
\def\Pr {{\rm Pr}}

\def\NP{\mathrm{NP}}
\def\P{\mathrm{P}}
\def\BPP{\mathrm{BPP}}
\def\cloud{\mathbf{cloud}}
\def\cH{\mathcal{H}}

\def\fin{\mathrm{fin}}

\providecommand{\keywords}[1]{\textbf{\textit{Keywords:}} #1}

\date{}

\begin{document}
\bibliographystyle{plain}

\title{Approximation algorithms for hitting subgraphs}

\author{Noah Br\"{u}stle~\thanks{School of Computer Science, McGill University. \texttt{noah.brustle@mail.mcgill.ca}} \and Tal Elbaz \thanks{School of Computer Science, McGill University. \texttt{tal.elbaz@mail.mcgill.ca}} \and  Hamed Hatami \thanks{School of Computer Science, McGill University. \texttt{hatami@cs.mcgill.ca}} \and  Onur Kocer  \thanks{School of Computer Science, McGill University. \texttt{onur.kocer@mail.mcgill.ca}} \and
Bingchan Ma  \thanks{School of Computer Science, McGill University. \texttt{bingchan.ma@mail.mcgill.ca}}}

\maketitle

\begin{abstract}
	Let $H$ be a fixed undirected graph on $k$ vertices.   The $H$-hitting set problem asks for deleting a minimum number of vertices from a given  graph $G$ in such a way that the resulting graph has no copies of $H$ as a subgraph. This problem is a special case of the hypergraph vertex cover problem on $k$-uniform hypergraphs, and thus admits an efficient $k$-factor approximation algorithm. The purpose of this article is to investigate the question that for which graphs $H$ this trivial approximation factor $k$ can be improved.  
\end{abstract}


\section{Introduction}

All graphs considered in this article are finite simple undirected graphs.   Given a fixed graph  $H$, a subset of the vertices of a graph $G$ is called an  \emph{$H$-hitting set} if it intersects every (not necessarily induced) copy  of $H$  in $G$. In other words, removing these vertices from $G$  results in an \emph{$H$-free} graph. The \emph{$H$-hitting set problem} asks for  the size of the smallest $H$-hitting set in a given graph $G$. When $H$ is a single edge, this is  the infamous \emph{vertex cover problem},  which is one of the most studied problems in the area of algorithmic graph theory.  Another closely related problem is  the \emph{feedback vertex set problem}, in which the goal is to remove a smallest set of vertices from $G$ so that the resulting graph contains no cycles. Note that here, instead of  a single graph $H$, we wish to eliminate a family of graphs, namely all cycles.  The vertex cover problem and the feedback vertex set problem are both NP-complete, however they both admit efficient $2$-factor approximation algorithms. 

The $H$-hitting set problem, as well as its  analogue for the induced subgraph setting, have been studied for   other specific graphs $H$ such as paths~\cite{MR3056999, MR2806015, MR2840533, MR3438294, MR3987192}, stars~\cite{MR3679310}, and cliques~\cite{MR3679310}.  It is not difficult to see that for any nonempty graph $H$, the $H$-hitting set problem is NP-complete (See Theorem~\ref{thm:NPcomplete} below).  On the other hand,  this problem is a special case of the hypergraph vertex cover problem for $k$-uniform hypergraphs where $k=|V(H)|$, and thus admits an efficient $k$-factor approximation.  That is, while there is a copy of $H$ in $G$, delete all the vertices of this copy, and repeat until the remaining graph becomes $H$-free. Since these detected copies of $H$  are all vertex-disjoint, any $H$-hitting set needs to remove at least one vertex from each copy.  Hence  the number of vertices that are removed by the algorithm is at most $k$ times the optimal solution. For the case of the vertex cover problem, it is widely believed that this simple algorithm is essentially optimal, in the sense that  for no fixed constant $\epsilon>0$, an efficient $(2-\epsilon)$-factor approximation algorithm  exists. In fact it is shown by Khot and Regev~\cite{MR2384079} that if the so called \emph{unique games conjecture} (UGC for short) is true, then  the existence of an efficient $(2-\epsilon)$-factor approximation  algorithm would imply $\P=\NP$. In fact their result overrules the existence of an efficient $(k-\epsilon)$-factor approximation algorithms for the $k$-uniform hypergraph vertex cover problem. This raises the following natural question. 

\begin{question}
\label{qu:approximate-easy}
For which graphs $H$ on $k$ vertices, there is a constant $\epsilon>0$ such that the $H$-hitting set problem admits an efficient  $(k-\epsilon)$-factor approximation algorithm? 
\end{question}

We shall refer to such graphs as \emph{approximate-easy}.  It is shown in~\cite{MR2894368} that there is an efficient $\frac{23}{11}$-factor approximation algorithm for the $P_3$-hitting set problem, where here and throughout the paper, $P_k$ denotes \emph{the path on $k$-vertices.} This was improved in \cite{MR2840533}  to a $2$-factor algorithm by showing that the $2$-factor primal-dual approximation algorithm of~\cite{MR1654290} for the feedback vertex set problem can be adapted to the $P_3$-hitting set problem. More recently it is shown in~\cite{Camby2014AP3} that the same ideas can be extended to give a $3$-factor approximation algorithm for the $P_4$-hitting set problem.  Lee~\cite{MR3987192} showed that for every $k$, there is an efficient  $O(\log(k))$-approximation algorithm for the $P_k$-hitting set, and in particular for sufficiently large $k$, the path $P_k$ is approximate-easy. Similarly, it is shown in \cite{MR3679310} that the star $S_k$, consisting of a vertex that is connected to $k$ other vertices, admits an  $O(\log(k))$-approximation algorithm, and thus    is approximate-easy provided that $k$ is sufficiently large. 

Let us now turn to negative results. The hardness of approximation for the $H$-hitting set problem has been studied extensively by Guruswami and Lee~\cite{MR3679310}. They prove that if $H$ is a $2$-vertex connected graph, then the $H$-hitting problem does not admit a $(k-1-\epsilon)$-approximation algorithms unless $\BPP \neq \NP$. Since Guruswami and Lee's goal was  not  to classify  the approximate-easy graphs, they preferred to focus on achieving the slightly weaker bound of $k-1-\epsilon$ and not rely upon the correctness of the UGC. However as they remark in their article, assuming the UGC conjecture, their approach can lead to the stronger $k-\epsilon$- bound that is relevant to our investigation.  

\begin{theorem}\cite{MR3679310}
\label{thm:hardnessGuruswamiLee}
Assuming the unique games conjecture and  $\NP  \not\subseteq \BPP$, no  $2$-vertex connected graph  is   approximate-easy. 
\end{theorem}

Since Theorem~\ref{thm:hardnessGuruswamiLee}  was claimed in~\cite{MR3679310} without a proof,   we present its proof in Section~\ref{sec:hardnessGuruswamiLee}.  The only graphs that are known to be approximate-easy are paths and stars. In the following theorem, we show that in fact all trees are approximate-easy. 
\begin{theorem}
\label{thm:trees}
Let $T$ be a tree on $k$ nodes. The $T$-hitting set problem admits an efficient $(k-\frac{1}{2})$-factor approximation algorithm.
\end{theorem}

Since our focus is only on the classification of approximate-easy graphs, we have not tried to optimize the approximation factor in Theorem~\ref{thm:trees}. The  proof of Theorem~\ref{thm:trees} can be applied to  a wider class of graphs. These are the graphs that  contain a vertex-cut that has certain properties (See Theorem~\ref{10}).  Inspired by these results and  Theorem~\ref{thm:hardnessGuruswamiLee} we conjecture the following. 

\begin{conjecture}
 $H$ is approximate-easy if and only if it is not $2$-vertex connected.
\end{conjecture}

The smallest example of an $H$ for which we do not yet have a definite answer is the graph consisting of a  triangle and a cycle of length $4$ that share a single vertex.

\paragraph{Hitting set problem for induced subgraphs:}  The hitting set problem can be defined analogously for induced subgraphs.  In this case, the goal is to  remove the minimum number of vertices from $G$, so that the remaining graph does not have any \emph{induced} copies of $H$. As in the case of the non-induced hitting set, there is a trivial $|V(H)|$-factor approximation algorithm, and thus one can analogously  define the notion of \emph{induced-approximate-easy}. In Section \ref{sec:induced} we show that the proof of Theorem~\ref{thm:hardnessGuruswamiLee} can be modified to imply a similar result for the induced case. 

\begin{theorem} 
\label{thm:InducedHardness}
Assuming the unique games conjecture and  $\NP  \not\subseteq \BPP$, if $H$ or its complement is  a  $2$-vertex connected graph, then $H$  is not induced-approximate-easy. 
\end{theorem} 

In particular $P_5$ as well as many other trees are \emph{not}  induced-approximate-easy, and thus the sets of approximate-easy and induced approximate-easy graphs are  distinct.

\section{The algorithms\label{sec:method}}

To develop our approximation algorithms, we need to  consider  the more general setting of the problem where $G$ is a vertex-weighted graph.  More precisely, we are given a graph $G$ where every vertex  has a non-negative weight, and the goal is to find the smallest possible total weight  among $H$-hitting sets in $G$. 

\paragraph{Phase I: Initial simplification using good subgraphs:}
Suppose that we are trying to develop a $t$-factor approximation algorithm for the $H$-hitting set problem. First note that we can remove the vertices with weight $0$ at no cost. The next important idea is the concept of \emph{$t$-good graphs} that is formally introduced in~\cite{MR4115652} but it is also implicit in some of the earlier algorithms. 

\begin{definition}
	\label{def:good}
A graph $K$ is called $t$-good for the $H$-hitting set problem if it is possible to assign non-negative weights to the vertices of $K$ such that every $H$-hitting set in $K$ has  weight at least $\frac{1}{t}$ of the total weight.  
\end{definition} 

In other words,  there is a choice of weights $w_K:V(K) \to \mathbb{R}^{\ge 0}$ for which, even picking all the vertices of $K$ is a $t$-factor approximation of the $H$-hitting set problem\footnote{Our notion of goodness is slightly stronger than that of \cite{MR4115652}  as we do not consider  minimality.}.   The key idea behind this notion is that if a weighted graph $(G,w_G)$ contains a copy of $K$ on the vertices with strictly positive weights, then we can make progress on $G$ in the following manner. With an abuse of notation let $w_K$ also denote the extension of $w_K$  to all the vertices of $G$ by assigning weight $0$ to the vertices that are not in that copy of $K$. Let $\lambda=\min_{v:w_K(v) \neq 0} \frac{w_G(v)}{w_K(v)}$, and let $w_1=\lambda w_K$ and $w_2=w_G-w_1$. Note that both $w_1$ and $w_2$ are non-negative functions, and furthermore $w_2$ assigns a weight of zero to at least one vertex in the copy of $K$. Let $S$ be a $t$-factor  $H$-hitting set  for $(G,w_2)$. Note that by the goodness property of $K$, $S$ is also automatically a $t$-factor $H$-hitting set for $(G,w_1)$. Since $w_G=w_1+w_2$, we conclude that $S$ is also a $t$-factor $H$-hitting set for $(G,w_G)$. This suggests the following approach. Let $\mathcal{K}={K_1,\ldots,K_\ell}$ be a set of $t$-good graphs for the $H$-hitting set problem.

\begin{algorithm}[H]
 \caption{Simplification of the problem using good subgraphs.\label{alg:RemoveGoods}}
 \KwData{On input $(G,w)$}
 \While{there is a copy of some $K \in \mathcal{K}$ in $G$ with strictly positive weights}{
  Set $\lambda=\min_{v: w_K(v) \neq 0} \frac{w(v)}{w_K(v)}$\;
  Replace $w$ with $w- \lambda w_K$ \;
 }
 Let $w_{\fin}$  denote the final weights.   
\end{algorithm}

Note that at every iteration of the algorithm, the weight of at least one more vertex of $G$  decreases to $0$, so the above algorithm terminates. It remains to find a $t$-factor $H$-hitting set for $(G,w_\fin)$.  Let $X$ be the set of the vertices that are assigned  weight  $0$ by $w_\fin$.  We can include the vertices of $X$ in a hitting set at no cost. Moreover $G-X$ is $K$-free for all $K \in \mathcal{K}$. Depending on $\mathcal{K}$, this could potentially restrict the structure of $G-X$ significantly, and allow us to find a $t$-factor $H$-hitting set $Y$ for $(G-X,w_\fin)$ efficiently. Then we can output $X \cup Y$ as a $t$-factor $H$-hitting set for $(G,w)$. 

\paragraph{Phase II: Improved factor based on colouring hypergraphs}
After the initial simplification in Phase I, we will end up with a weighted graph $G$ that is $K$-free for all graphs $K \in \mathcal{K}$, where $\mathcal{K}$ is our set of good graphs. In the second phase, we will use the $\mathcal{K}$-freeness of $G$ to find a desired colouring of the vertices of $G$ that enables us to solve the hitting set problem efficiently with an approximation factor that is strictly less than $|V(H)|$. This is based on some known results for approximating the hypergraph vertex cover problem as described below. 

Consider the following setting. Let $\cH=(V,E,w)$ be a vertex-weighted hypergraph where every edge is of size at most $k$, and $w:V \to \mathbb{R}^+$. The minimum vertex cover problem in this setting is the solution to the following integer linear program: 
$$
\begin{array}{clcl}
\min & \sum_{v \in V} w(v) x_v &\qquad &\\
{\rm s.t.} &  \sum_{v:v \in e} x_v \ge 1 & & \forall e \in E\\
& x_v \in \{0,1\}  & & \forall v \in V
\end{array}
$$
 Let us denote the solution to this problem as $\tau(\cH)$. We can relax this to a linear program
$$
\begin{array}{clcl}
\min & \sum_{v \in V} w(v) x_v &\qquad &\\
{\rm s.t.} &  \sum_{v:v \in e} x_v \ge 1 & & \forall e \in E\\
& x_v \ge 0  & & \forall v \in V
\end{array}
$$
Let $\tau^*(\cH)$ denote the cost of the optimal solution to this linear program. This is known as the fractional cover number of $\cH$. Finally note that by the linear program duality, this is equal to the solution to the following linear program, which solves the maximal fractional matching problem. 
$$
\begin{array}{clcl}
\max & \sum_{e \in E} y_e &\qquad &\\
{\rm s.t.} &  \sum_{e:v \in e } y_e  \le w(v) & & \forall v  \in V\\
& y_e \ge 0  & & \forall e \in E
\end{array}
$$

\begin{definition}
A hypergraph $\cH$ is called $t$-colourable if there exists a $t$-colouring of the vertices of $\cH$ such that every edge of size at least $2$ contains at least $2$ different colours.  
\end{definition}
The following theorem is adapted from  Aharoni et al ~\cite{MR1401890}, who demonstrated a bound on the ratio of $\tau$ and $\tau^*$ for $t$-colourable hypergraphs. We modify their arguments to present an explicit efficient approximation algorithm for  $\tau$, and furthermore generalize it to the case of weighted vertices.
\begin{theorem}
\label{thm:Aharoni}
Let $H$ be a graph on $k$ vertices, and $t \ge k$ be an integer.  There is an efficient $k(1-\frac{1}{t})$-factor approximation algorithm that solves the $H$-hitting set problem for weighted graphs $G$ that admit a $t$-colouring such that there are no monochromatic copies of $H$ in $G$. 
\end{theorem}

\begin{proof}

Suppose that we have a $t$-colouring of $G$, and let $\cH$ be the hypergraph defined by the vertices of $G$ and the edges $e \in E$ corresponding to the copies of $H$ in $G$. We will show that we are able to find a set of vertices of total weight at most $k(1-\frac{1}{t})\tau(\cH)$ that covers $\cH$.

Let $V$ be the set of vertices in $G$ which are contained in some copy $H_i$ of $H$ in $G$, and let $\cH'=(V,E)$ be the hypergraph defined on $V$ with the   hyperedges by the copies of $H$. Let $g: V \to \mathbb{R}^+$ be a minimal fractional cover of $\cH'$ and $f:E \to \mathbb{R}^+$ be a maximal fractional matching in $\cH'$ with values 
$| g | = | f | = {\tau^*}$.

We consider two cases. 

\begin{enumerate}
    \item $g(v) > 0$ for every $v \in V$.
    
    According to the complementary slackness conditions,
$$w(V) = \sum_{v \in V}w(v) = \sum_{v \in V}\sum_{e\ni v} f(e) = \sum_{e \in E} f(e)| e \cap V | \leq \sum_{e \in E} f(e)k = k \sum_{e \in E} f(e) = k \tau^*(H), $$
and thus
\begin{equation}
\label{eq:1}
\tau^*(H) \geq \frac{w(V)}{k} .
\end{equation}

On the other hand, the union of any $(t-1)$ colours of H is obviously a cover of H, so 
\begin{equation}
\label{eq:2}
\tau(H) \leq \frac{t-1}{t} w(V).
\end{equation}

Let $S$ be a uniformly coloured set of vertices in $V$ of maximal size.
We may then obtain a discrete $(1- 1/t) k$ approximation of the fractional $H$-Hitting Set Problem by choosing the vertices $V\setminus S$.
    
    \item There exists a $v \in V$ such that $g(v) = 0$.
    
    We argue by induction on the number of vertices in  $G$. Say $e$ is a copy of H containing $v$.
    
    Our base cases are the graphs for which case 1 applies; as the empty set holds this condition, we will eventually arrive at such a case by reducing the number of vertices.

    Necessarily, as $\sum_{v \in v} g(e) \geq 1$, we must have a vertex $v'$ in $e$ such that $g(v')> \frac{1}{k-1}$.

Note that $g$ restricted to $V\setminus v'$ is clearly a valid fractional covering of $G\setminus v'$ (any $e$ in $G\setminus v'$ must clearly also be covered in $G$). Clearly, we also have a valid $t$-colouring of $G\setminus v'$. Thus, we may obtain a $(1- 1/t) k$ approximation of the fractional $H$-Hitting Set Problem on $G\setminus v'$ by our induction hypothesis, which will have total weight no greater than $(1- 1/t) k\sum_{v \in V\setminus v'}w(v)$. Let $J$ be the set of vertices selected in this manner.

Further, we give $v'$ weight $1$. Note that $g(v') \geq \frac{1}{k-1} \geq \frac{1}{(1- 1/t)k}$, thus $|J\cup v'| \leq (1- 1/t)k \tau^*$.

Since all $e$ in $G$ either will be covered by $v'$ (if they contain $v'$), or by $J$ (if they are in $V \setminus v'$), this is a valid discrete $(1- 1/t)k$-approximation of the fractional $H$-hitting set problem.

Necessarily, as the size of the vertex set of $G$ decreases by 1 at each iteration, we eventually have either an empty graph (which is trivial to approximate) or a graph such that $g(v)>0$ for all vertices $v$ (which we know how to approximate). Thus, our algorithm terminates with a valid approximation.

\end{enumerate}

To complete our proof of the theorem, simply note that $\tau \geq \tau*$, and thus, the algorithm that we have used equally finds us a valid $(1- 1/t)k$-approximation of the discrete $H$-hitting set problem on $G$.
\end{proof}

This gives us the following approximation algorithm:

\begin{algorithm}[H]
 \caption{ColorSimp}
 \KwData{On input $(G,H,c)$}
 
 Let $V$ be the set of vertices of $G$ within a copy of $H$ in $G$
 
 Let $G'$ be the hypergraph defined on $V$ with edges $e \in E$, the copies of $H$ in $G$

 Let $g: V \to \mathbb{R}^+$ be a minimal fractional cover of $G'$

 \eIf{$g(v)>0, \forall v \in V$}{
 Let $J = \argmax_{S \subseteq V \mid c(s) = c(r) \forall s, r \in S} |S|$
 
 \Return $V \setminus J$
 }{
 Choose  $v \in V, g(v)=0$
 
Choose $e \in E \mid v \in e$

Let $v' = \argmax_{u \in e} (g(v))$

Let $G^*$ be the graph $G$ without $v'$ and any of its adjacent edges.

\Return $v' \cup ColorSimp(G^*, H, c)$
}

\end{algorithm}

Our approach for designing approximation algorithms for the $H$-hitting set problem is to   find a set $\mathcal{K}$ of good graphs for $H$ such that every $\mathcal{K}$-free graph $G$ admits a colouring (that can be found efficiently). Then we can apply Phase I to the initial graph to simplify the graph to a $\mathcal{K}$-free graph, and then apply Theorem~\ref{thm:Aharoni} to obtain a desired $H$-hitting set.

\subsection{The approximate-easy graphs}

In this section, we will apply the method that was developed in Section~\ref{sec:method} to establish that trees are approximate-easy. Our proof implies that a broader class of graphs are approximate-easy. In order to define this class, we need to introduce the notion of a semi-symmetric cut vertex.

\begin{definition}
\label{def:semi}
Let $H$ be a graph consisting of $m$ connected graph $F_1,\ldots,F_r$, all sharing a single vertex $v$, and otherwise having distinct vertices. We call $v$ a \emph{semi-symmetric cut-vertex} of $H$ if there exists distinct $i,j \in  [r]$ such that $F_i$ is a subgraph of $F_j$ as $v$-rooted graphs. 
\end{definition}

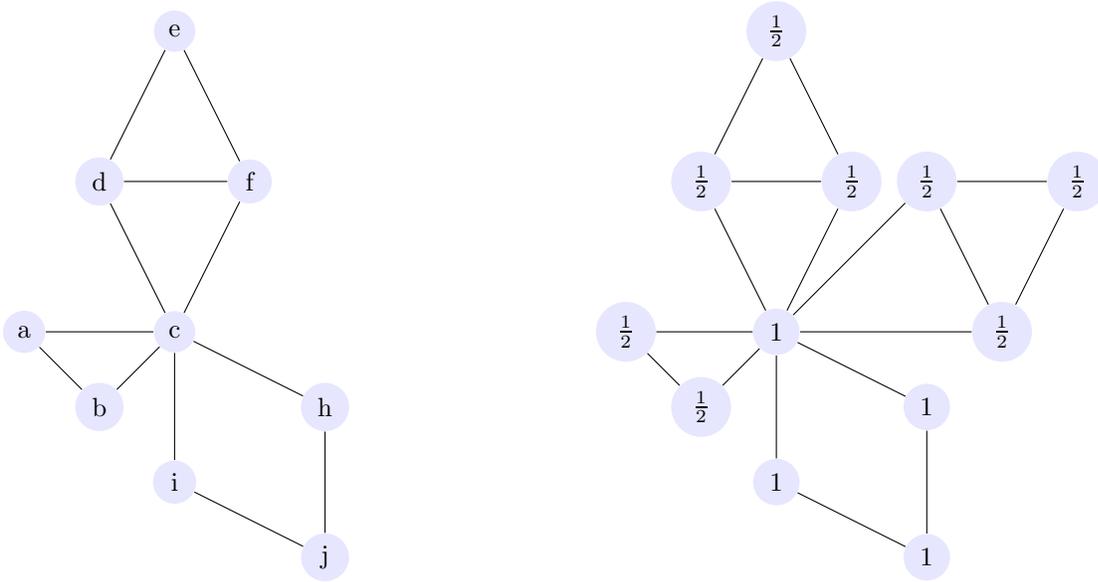
\begin{figure}
\begin{tikzpicture}
      [scale=1,auto=left,every node/.style={circle,fill=blue!10}]
      \node (n1) at (1,4) {a};
      \node (n2) at (2,3)  {b};
      \node (n3) at (2,6)  {d};
      \node (n4) at (3,2) {i};
      \node (n5) at (3,4)  {c};
      \node (n6) at (3,8)  {e};
      \node (n7) at (4,6)  {f};
      \node (n8) at (5,1)  {j};
      \node (n9) at (5,3)  {h};
      \node (m1) at (9,4) {$\frac{1}{2}$};
      \node (m2) at (10,3)  {$\frac{1}{2}$};
      \node (m3) at (10,6)  {$\frac{1}{2}$};
      \node (m4) at (11,2) {$1$};
      \node (m5) at (11,4)  {$1$};
      \node (m6) at (11,8)  {$\frac{1}{2}$};
      \node (m7) at (12,6)  {$\frac{1}{2}$};
      \node (m8) at (13,1)  {$1$};
      \node (m9) at (13,3)  {$1$};
      \node (m10) at (13,6)  {$\frac{1}{2}$};
      \node (m11) at (14,4)  {$\frac{1}{2}$};
      \node (m12) at (15,6)  {$\frac{1}{2}$};
    
      \foreach \from/\to in {n1/n2,n2/n5,n5/n1,n3/n6,n6/n7,n7/n3,n3/n5, n7/n5, n4/n5, n5/n9, n8/n4, n8/n9, m1/m2,m2/m5,m5/m1,m3/m6,m6/m7,m7/m3,m3/m5, m7/m5, m4/m5, m5/m9, m8/m4, m8/m9, m5/m10, m5/m11, m10/m11, m11/m12, m12/m10}
        \draw (\from) -- (\to);
        
\end{tikzpicture}
\caption{
\label{fig:semi}
On the left: $c$ is a semi-symmetric cut-vertex in a graph $H$ as the $c$-rooted subgraph induced by $(c, d, e, f)$ contains the $c$-rooted subgraph induced by $(c, d, f)$ as a subgraph. 
On the right: This graph is at least $k-\frac{1}{2}$ good for $H$ where $H$ is the graph displayed on the left side.}
\end{figure}

\begin{theorem}
Every graph $H$ containing a semi-symmetric cut-vertex is approximate-easy. More precisely, there is an efficient  $(|V(H)| - \frac{1}{2})$-factor approximation algorithm for the $H$-hitting set problem.
\label{10}
\end{theorem}

\begin{proof}
Let $v$ be a semi-symmetric cut-vertex in $H$, and let $F_1,\ldots,F_r$ be as in Definition~\ref{def:semi}, and $F_i$   be a subgraph of $F_j$ as $v$-rooted graphs. Denote $k=|V(H)|$. Construct  a new graph $H'$ from $H$ by attaching an additional copy of $F_j$ to $v$, say $F_j'$. We will show that $H'$ is $\left(k-\frac{|F_i|-1}{2}\right)$ good for $H$. This can be easily verified by assigning a weight of $\frac{1}{2}$ to all vertices in $V(K_i) \cup V(K_j) \cup V(K_j') \setminus \{v\}$, and a weight of $1$ to all other vertices in $H$. Note that every $H$-hitting set in $H'$ either   includes one of the vertices with weight $1$, or at least two of the vertices with weight $\frac{1}{2}$. Since the total weight is $k-\frac{|F_i|-1}{2}$,  we conclude that $H'$ is $(k-\frac{|F_i|-1}{2})$-good for $H$.

Next we run the algorithm in Section~\ref{sec:method} with $H'$ as the only good graph. We will arrive at an $H'$-free  weighted graph $G'$. It remains to obtain a $(k- \frac{1}{2})$-factor approximation of the $H$-hitting set problem for $G'$.

We say that a copy  of $H$ in $G'$  is centred at $u \in V(G')$ if $u$ can correspond to $v$ in this copy of $H$. We call $u \in V(G')$ central if some copy of  $H$  is centred at $u$. Similarly we say that a copy of $F_j$ is centred at $u$ if $u$ can correspond to $v$ in $F_j$. 

For each central vertex $u$ in $G'$, select an arbitrary copy $H_u$ of $H$ centred at $u$. Let $S_u$ be the set of vertices inside $H_u$ excluding $u$.   Every copy of $F_j$ centred at $u$ must intersect $S_u$, since otherwise we would be able to extend $H_u$ to a copy of $H'$ in the $H'$-free graph $G'$. 

We now construct a directed graph $D$ with vertices $V(G')$, and the directed edges $(u,w)$ for any central vertex $u$, and every $w \in S_u$. Note that $D$ has maximal out-degree $k-1$.

\begin{lemma}
Every directed graph of maximum out-degree $m$ admits a proper $(2m+1)$-colouring. 
\end{lemma}
\begin{proof}
Since the sum of the in-degrees is equal to the sum of the out-degrees, such a graph must contain at least one vertex with total degree at most $2m$. We can remove this vertex, colour the rest of the graph inductively, and then colour this vertex with one of the available colours. 
\end{proof} 

We may thus colour $D$ with at most $2(k-1)+1 < 2k$ colours. Colour all vertices in $G'$ accordingly. This is a valid $2k$-hypergraph colouring of the vertices in $G'$ with hyperedges corresponding to  the copies of $H$ in $G'$.  Every copy $H_0$ of $H$ contains a central vertex $u$. Since there might be other copies of $H$ centred at $u$, $H_0$ might not be the copy  used to define $S_u$, however it still contains a copy of $F_j$ centred at $t$, and thus  it has at least one element in $S_u$.  This vertex is coloured differently than $t$, and thus  $H_0$ is at least 2-coloured. 

Applying Theorem~\ref{thm:Aharoni}, we obtain a $k-\frac{k}{2k} = k- \frac{1}{2}$ approximation of the hitting set problem for the copies of $H$ in $G'$.
\end{proof}

\begin{corollary}
Every graph $H$ with at least three  vertices and at least one vertex of degree 1 is approximate-easy, and has an approximation factor of at most $k - \frac{1}{2}$. 
\end{corollary}

\begin{proof}
Let $u$ be a vertex of degree 1. Let $v$ be adjacent to $u$. If $H$ has at least three vertices, $v$ must be of minimum degree 2; $v$ is necessarily a semi-symmetric cut vertex  with the single edge $uv$ as our choice of $F_i$, and any other adjacent component as $F_j$. We may apply the previous theorem.
\end{proof}

\begin{corollary}
Every tree $T$ on at least 3 vertices is approximate-easy, and has an approximation factor of at most $k - \frac{1}{2}$.
\end{corollary}

\section{Hardness results}

In this section, we present our results regarding the hardness of the $H$-hitting set problem. In Theorem~\ref{thm:NPcomplete} below, we prove that unless $H$ is an empty graph, the $H$-hitting set problem is NP-complete. We note in Theorem~\ref{thm:NPcompleteInduced} that the same argument implies the NP-completeness of the \emph{induced} $H$-hitting set problem for every $H$. Next in  Section~\ref{sec:hardnessGuruswamiLee} we present the proof of Theorem~\ref{thm:hardnessGuruswamiLee} by using a simplified version of  Guruswami-Lee's~\cite{MR3679310} argument to show that $2$-vertex connected graphs are not approximate-easy. Finally, in Section~\ref{sec:induced}, we prove Theorem~\ref{thm:InducedHardness}, by adapting the proof of Theorem~\ref{thm:hardnessGuruswamiLee} to the induced setting to show that if $H$ or its complement is $2$-vertex connected, then $H$ is not \emph{induced} approximate-easy.

\subsection{NP-completeness} 
\begin{theorem}
\label{thm:NPcomplete}
The $H$-hitting set problem is NP-complete for every connected graph $H$ with at least two vertices.
\end{theorem}

 Theorem~\ref{thm:NPcomplete} follows immediately from Lemma~\ref{lem:degree2} and Lemma~\ref{lem:degree1} below.

\begin{lemma}
\label{lem:degree2}
The $H$-hitting set problem is NP-complete if  $H$ is a connected graph with minimum degree $2$,
\end{lemma}
\begin{proof}
The proof is by a reduction from the Vertex Cover Problem. Let $G$ be an input to the vertex cover problem. 

Every connected graph can be uniquely decomposed into a tree of its maximal 2-connected components, called block-cut tree. Let $J$ be a 2-connected subgraph corresponding to a leaf of the block-cut tree of $H$. Note that $J$ contains at most one cut-vertex, and since the minimum degree is $2$, there must be an edge $e_0$ in $J$ such that neither of the endpoints of $e$ is a cut-vertex in $H$. 

 We construct a graph $G'$ by ``gluing'' a copy $H_{e}$ of $H$ onto every edge $e$ of $G$ via the edge $e_0$: More precisely, we take the disjoint union of the two graphs and identify\footnote{We arbitrarily choose a start and an end point for both $e_0$ and $e$, and identify the starts together, and the ends together.} the two edges $e$ and $e_0$. We will show that solving the $H$-hitting problem on $G'$ allows us to solve the vertex cover problem on the original graph $G$.

First note that every vertex cover $S$ for $G$ is  an $H$-hitting set  for $G'$ as  removing the vertices in $S$ from $G'$ eliminates all the edges in $G$, and moreover if $u$ is a vertex in $V(G) \setminus S$, then since neither of the endpoints of $e$ was a cut-vertex, $u$ cannot belong to any copy of $H$ in $G' - S$.

For the other direction, consider an $H$-hitting set $T$ in $G'$, and let $S=T \cap V(G)$. Let $E \subseteq E(G)$ be the set of the edges in  $G$ that are not covered by $S$. Note that for every $e\in E$, $T$ must contain at least one vertex from $H_e$. Hence $|T| \ge |S| + |E|$, and the latter is obviously an upper bound on the size of a minimum vertex cover  for $G$. 

We conclude that the size of a minimum vertex cover in $G$ is equal to the size of the smallest $H$-hitting set in $G'$.
\end{proof}

Next in Lemma~\ref{lem:degree1}  we  establish NP-completeness for the case where $H$ contains a vertex of degree $1$. 
\begin{lemma}
\label{lem:degree1}
The $H$-hitting set problem is NP-complete if  $H$ is a connected graph with minimum degree $1$,
\end{lemma}
\begin{proof}
Again the proof is by a reduction from the vertex cover problem. Let $G$ be an instance of the vertex cover problem. Let $v_0$ be a vertex of degree $1$ in $H$, and let $u_0$ be the unique neighbour of $v_0$. Let $F=H-v_0$.

This time, we obtain a graph $G'$ by gluing a copy of $F$ on every \emph{vertex} of $G$ via the vertex $u_0$. More formally,  for every vertex $u$ of $G$, we add a disjoint copy $F_u$ of $F$, and unify $u$ and $u_0$.

Let $S$ be a vertex cover for $G$.  Removing $S$ from $G'$ turns it into a disjoint union of copies of $F$, which is $H$-free. Thus $S$ is an $H$-hitting set for $G'$.

For the other direction, consider an $H$-hitting set $T$ in $G'$, and let $S=T \cap V(G)$. Let $R \subseteq V(G)$ be the set of the vertices in  $G$ that are involved in the remaining edges in  $G-S$. Note that $S \cup R$ is a vertex cover for $H$. For every $u\in R$, $T$ must contain at least one vertex from $H_u$. Hence $|T| \ge |S \cup R|$.

We conclude that the size of smallest $H$-hitting set in $G'$ is equal to the size of a minimum vertex cover in $G$. 
\end{proof}

The proofs of Lemma~\ref{lem:degree1}~and~\ref{lem:degree2} apply to the induced case as well. We conclude that Theorem~\ref{thm:NPcomplete}  also holds for the induced $H$-hitting set problem. 

\begin{theorem}
\label{thm:NPcompleteInduced}
The induced $H$-hitting set problem is NP-complete for every connected graph $H$ with at least two vertices.
\end{theorem}

\subsection{Guruswami-Lee's Hardness Result: Theorem~\ref{thm:hardnessGuruswamiLee}} 

\label{sec:hardnessGuruswamiLee}

In this section, we present Guruswami-Lee's hardness of approximation result albeit with minor modifications. The starting point is the hardness of the hypergraph vertex cover.  

\begin{theorem} \cite{MR2384079}
\label{thm:BansalKhot}
Fix an integer $k>2$, and let $\epsilon \in (0,1)$.  Given a $k$-uniform hypergraph $\cH=(V_\cH,E_\cH)$, 
assuming the UGC and $P \neq \NP$, there is no polynomial time algorithm that distinguishes the following cases.
\begin{itemize}
 \item {\bf Completeness:} There exist disjoint subsets $V_1,\ldots,V_k \subseteq V_\cH$, each with $\frac{1-\epsilon}{k}$ fraction of vertices, such that each hyperedge has at \emph{most} one vertex in each $V_i$. Note that in this case, every $V_i$ together with the vertices in $V_0:=V_\cH \setminus (V_1 \cup \ldots \cup V_k)$ form a vertex cover with $(\frac{1-\epsilon}{k}+ \epsilon)$ fraction of vertices.

 \item {\bf Soundness:} Every subset of $V_\cH$ with a less than $(1-\epsilon)$ fraction of vertices does not intersect at least one hyperedge. Equivalently, every subset $C$ of $\epsilon$-fraction of vertices wholly contains a hyperedge. 
\end{itemize}
\end{theorem}

In order to deduce a hardness result for the $H$-hitting set problem from Theorem~\ref{thm:BansalKhot}, naturally one would think of  replacing each hyperedge $e=(v_1,\ldots,v_k)$ of $\cH$ with a copy of $H$. However, this can lead to a problem as one might create unintentional copies of $H$ that come from a combination of different  hyperedges. Thus a vertex cover for $\cH$ might not necessarily correspond to an $H$-hitting set for this graph. To overcome this problem, we will replace each vertex of $\cH$ with a large ``cloud'' of vertices, and for each hyperedge $e=(v_1,\ldots,v_k)$ we randomly implant several copies of $H$ on the clouds of these vertices.

\begin{theorem}
\label{thm:hardness}
Let $H$ be a $2$-vertex connected graph on $k$ vertices, and let $\epsilon>0$. Assuming the UGC, unless $\NP  \subseteq \BPP$, no polynomial time algorithm  can distinguish between the following two cases for a graph $G$. 
\begin{itemize}
 \item Completeness: There is an $H$-hitting set with $\frac{1}{k}+ \epsilon$ fraction of the vertices. 
 \item Soundness: Every set with $2\epsilon$ fraction of the vertices contains at least a copy of $H$.  
\end{itemize}
In particular no efficient algorithm can  approximate the $H$-hitting set problem with a constant factor that is strictly less than $k$. 
\end{theorem}
\begin{proof}
Let $\cH$ be the $k$-uniform hypergraph from Theorem~\ref{thm:BansalKhot}.  We will construct a polynomial size random graph $G$ such that with probability at least $7/8$, it will satisfy the following property: approximating the $H$-hitting set problem on $G$ would distinguish the two cases in Theorem~\ref{thm:BansalKhot}. Since $G$ is randomly constructed, we can only conclude the hardness result under the assumption that  $\NP  \not\subseteq \BPP$.

Without loss of generality we will assume $V_H=[k]$.  Given the $k$-uniform  hypergraph $\cH$,  we put an arbitrary order on the vertices of every hyperedge $e=(v_1,\ldots,v_k)$ of $\cH$.

We may assume that $n=|V_\cH|$ is sufficiently large as a function of $\epsilon$ and $k$, as otherwise the hypergraph vertex cover problem on $\cH$ could be solved efficiently. Let $B=B(n,\epsilon,k)$ be a sufficiently large number that polynomially depends on $n$ to be determined later.
Let $\lambda=\lambda(k,\epsilon)$ be a positive integer to be determined later as well. The random graph $G$ is defined in the following manner: 
\begin{itemize}
 \item  $V_G=V_\cH \times [B]$. That is we replace every vertex $v$ of the hypergraph $\cH$ with $B$ new vertices. We refer to $\cloud(v):=\{v\} \times [B]$ as the  \emph{cloud} of $v$.
 
 \item For every edge $e=(v_1,\ldots,v_k) \in E_\cH$, we plant $\lambda B$ copies of $H$ in $G$. For $j=1,\ldots,\lambda B$ repeat:
 \begin{itemize}
 \item  Pick  $(\ell_1,\ldots,\ell_k) \in [B]^k$ uniformly at random and plant a copy of $H$ on $(v_1,\ell_1),\ldots,(v_k,\ell_k)$ by mapping the   $i$-th vertex of $H$ to $(v_i,\ell_i)$.  We put a tag of $[e,j]$ on all the edges of this copy of $H$.       
 \end{itemize}
  
\end{itemize}

The above procedure produces a random graph $G$, together with a map $\psi:E_G \to E_\cH \times [\lambda B]$ representing the tag of each edge in $G$. 
Note that $G$ can have multiple edges between two vertices. While we can replace these edges with a single edge without affecting  the set of $H$-hitting sets, for the sake of the presentation, it will be convenient  to keep them as multiple edges.  

We refer to the planted copies of $H$ in $G$ as \emph{intended} copies. Note that a copy of $H$ in $G$ is intended if and only if all the edges in the copy have the same tag.  However $G$  might also have other copies of $H$, which we refer to as \emph{unintended}.

\paragraph{Completeness:} Suppose that $\cH$ satisfies the conditions of the completeness case of Theorem~\ref{thm:BansalKhot}. Let $S = (V_0 \cup V_1) \times [B]$. Since $(V_0 \cup V_1)$ is a vertex cover in $\cH$, $S$ hits every intended copy of $H$. We will show that with probability at least $3/4$, there will be only few unintended copies of $H$ that do not intersect $S$. Consequently, we can hit those copies by adding few extra vertices to $S$. Consider an unintended copy of $H$ in $G$ given by a map $\phi:[k] \to V_G$. Since this copy is unintended, there are   $p>1$ different tags $t_1,\ldots,t_p$ on its edges. Let $I_i \subseteq [k]$ be the set of the vertices of $H$ that are incident to the edges with the  tag $t_i$ in that copy.  Since $H$ is $2$-vertex connected, each $I_i$ has at least two vertices that belong to some other $I_j$ as well. This implies  $|I_1|+\ldots+|I_p| \ge k + p$.  

There are  $(nB)^k$ choices for $\phi$, and fixing $p$, there are at most $p^{|E(H)|} (\lambda B |E_\cH|)^p$ choices for the tags on the edges of this copy of $H$. For a fixed $\phi$ and fixed tags, the probability that the corresponding tagged copy of $H$ is in $G$ is $B^{-(|I_1|+\ldots+|I_p|)}  \le  B^{-k-p}$. We conclude that the expected number of unintended copies of $H$ is at most  
$$\sum_{p=2}^{|E(H)|} (nB)^k p^{|E(H)|} (\lambda B |E_\cH|)^{p}  B^{-k-p}  \le \lambda^{k^2} n^{k^2}.$$
Taking $B=\lambda^{k^2} n^{k^2}$, we see that the expected number of unintended copies of $H$ is  $B$, which is very small compared to  $|V(G)|=nB$.  Assuming  $n>\frac{k}{\epsilon}$ and applying Markov's inequality, the probability that there are more than $4 B \le \frac{\epsilon}{k} n B =\frac{\epsilon}{k} |V(G)|$ unintended copies of $H$ in $G$ is at most $\frac{1}{4}$. Thus, with probability at least $\frac{3}{4}$, there is an $H$-hitting set in $G$ of size at most 
$$|S|+4 B \le \left(\frac{1-\epsilon}{k} + \epsilon \right)|V_G| + \left(\frac{\epsilon}{k}\right) |V_G| =  \left(\frac{1}{k} + \epsilon\right) |V_G|,$$
as desired.

\paragraph{Soundness:}  Next suppose that $\cH$ satisfies the conditions of the soundness case of Theorem~\ref{thm:BansalKhot}. First we show that with probability at least $7/8$, the random graph $G$ satisfies the following property: For every edge $e=(v_1,\ldots,v_k) \in E_\cH$, for any choice of subsets $A_i \subseteq \cloud(v_i)$ with  $|A_i| \ge \epsilon B$ for all $i \in [k]$, there is a copy of $H$ on the induced subgraph of $G$ on $A_1 \cup \ldots \cup A_k$. Indeed, for any choice of $e$ and $A_i$'s, the probability that none of the $\lambda B$ planted copies of $H$ that are created by $e$ fall into this set is at most 
$$(1-\epsilon^k)^\lambda \le e^{-\lambda B \epsilon^k}.$$ 
Applying a union bound over $e$ and $A_i$'s, we can bound this probability by 
$$|E_\cH| 2^{kB} e^{-\lambda B \epsilon^k} \le 1/8,$$
for $\lambda \ge  k \epsilon^{-k}$.

Now consider a $G$ that satisfies the above property, and let $D \subseteq V_G$ be a set with at least $2\epsilon$ fraction of the vertices. Let $C \subseteq V_\cH$ be the set of  vertices $v \in \cH$  such that $|\cloud(v) \cap D| \ge \epsilon B$. Since $|D| \ge 2 \epsilon |V(G)|$, we know that $|C| \ge \epsilon |V_\cH|$, and thus there is a hyperedge $e$ in $C$. Consequently, there is a copy of $H$ in $D$.

\end{proof}

\subsection{Hitting sets for induced subgraphs: Proof of Theorem~\ref{thm:InducedHardness}} \label{sec:induced}
In this section we present the proof of  Theorem~\ref{thm:InducedHardness} by proving an analogue of Theorem~\ref{thm:hardness} for induced hitting sets. 

\begin{theorem}
Let $H$ be a $2$-vertex connected graph on $k$ vertices. Assuming the UGC, unless $\NP  \subseteq \BPP$, no polynomial time algorithm  can distinguish between the following two cases for a graph $G$. 
\begin{itemize}
 \item Completeness: There is an induced $H$-hitting set with $\frac{1}{k} + \epsilon$ fraction of the vertices. 
 \item Soundness: Every set with $3\epsilon$ fraction of the vertices contains at least one induced copy of $H$.  
\end{itemize}

In particular no efficient algorithm can  approximate the induced $H$-hitting set problem with a constant factor that is strictly less than $k$. 
\end{theorem}

\begin{proof}
Create the random graph $G$ precisely as in the proof of Theorem \ref{thm:hardness}.  Therefore, we can see that there are $|E_{\mathcal{H}}| \lambda B$ intended copies of $H$. However, some of these copies might not remain \emph{induced} copies of $H$ due to possible intersections with other intended copies.  We call an intended copy of $H$ in $G$ \emph{destroyed}  if it is \emph{not} an \emph{induced} copy of $H$. Note that this happens exactly when another intended copy of $H$ plants an edge between two vertices that are not supposed to be connected in this copy. As it is explained below, the proof follows by showing that with high probability the number of destroyed copies is small.

\paragraph{Completeness:} In Theorem \ref{thm:hardness},  it was proven that with probability at least $\frac{3}{4}$, there is an (not necessarily induced) $H$-hitting set of size at most  $(\frac{1}{k}+\epsilon) |V(G)|$ in $G$.   Since an  $H$-hitting set  is also an induced  $H$-hitting set, the completeness follows.

\paragraph{Soundness:} In this case, we need to show that the number of destroyed intended copies of $H$ is small.  For   $e \in E_{\mathcal{H}}$ and $i \in \{1,2,\ldots,\lambda B\}$, let $H_{e,i}$ denote the $i$-th intended copy of $H$ in $G$ arising from $e$.  Note that for $H_{e,i}$ to be destroyed, there must be another pair $(e',j)$ such that   $e' \neq e$, and $H_{e,i}$ and $H_{e',j}$ intersect in at least two vertices. Note that  $H_{e,i}$ cannot be destroyed by another $H_{e,j}$. 

 Let the random variable $X$ denote the number of destroyed copies. From the above discussion, $X$ is obviously bounded by the number of $(e,i,e',j)$ such that $e \neq e'$ and  $H_{e,i}$ and $H_{e',j}$ intersect in at least two vertices. Hence by linearity of expectation
 $$\Ex[X] \le \sum_{e \neq e'} \sum_{i,j} \Pr[|V(H_{e,i}) \cap V(H_{e',j})| \ge 2].$$
 
For fixed $e,e',i,j$,  in order to have $|V(H_{e,i}) \cap V(H_{e',j})| \ge 2$,  the hyperedges $e$ and $e'$ must intersect in at least two vertices $u,v \in V_{\mathcal{H}}$, and moreover $H_{e,i}$ and $H_{e',j}$ must have landed on the same vertices in $\cloud(u)$ and $\cloud(v)$. There are at most $k^2$ choices for $u$ and $v$, and given $u$ and $v$, the probability that these copies land on the same vertices on both clouds is exactly $1/B^2$. Hence by applying the union bound on all the possible choices of $u,v \in e \cap e'$, we have  $\Pr[|V(H_{e,i}) \cap V(H_{e',j})| \ge 2] \le k^2/B^2$. We conclude that 
\begin{equation}
    \mathbb{E}[X]\leq|E_{\mathcal{H}}|^2 (\lambda B)^2 \frac{k^2}{B^2} = |E_{\mathcal{H}}|^2 \lambda^2 k^2.
\end{equation}

Now, using Markov's inequality, the probability that more than $10 |E_{\mathcal{H}}|^2 \lambda^2 k^2 $ intended copies are destroyed is at most $\frac{1}{10}$. Thus with probability at least $\frac{9}{10}$, the number of the vertices that are involved in   destroyed copies of $H$ is at most $ k  \times (10 |E_{\mathcal{H}}|^2 \lambda^2 k^2)\le \epsilon |V(G)|$. Now consider a subset of $V(G)$ of size at least  $3\epsilon |V(G)|$. Then  $2\epsilon |V(G)|$ of these vertices are not in any destroyed copies, and thus by the proof of Theorem~\ref{thm:hardness}, they contain an intended copy of $H$. This copy is induced as it is not part of any destroyed copy.

\end{proof}

 \bibliographystyle{alpha}
\bibliography{hittingset}

\end{document}